\documentclass[runningheads]{llncs}
\usepackage{hyperref}
\usepackage{listings}
\usepackage{enumitem}
\usepackage{mathtools}
\usepackage{amsmath,amssymb,amsfonts}
\usepackage{mathrsfs,mathabx,verbatim,fullpage,psfrag,graphicx,color}

\usepackage{amsthm}
\usepackage{cleveref}
\usepackage{algorithmicx}
\usepackage{fixltx2e}
\usepackage[noend]{algpseudocode}
\MakeRobust{\Call}

\newcommand{\icode}[1]{\texttt{\small #1}}

\newcommand{\coll}{\mathsf{coll}}
\newcommand{\cycl}{\mathsf{cycl}}

\newcommand{\distinct}{\mathrm{distinct}}
\newcommand{\permutation}{\mathrm{perm}}

\newcommand{\eg}{{\it e.g.} }
\newcommand{\ie}{{\it i.e.} }

\title{$k$-Equivalence Relations and Associated Algorithms}
\author{Daniel Selsam \and Jesse Michael Han}
\authorrunning{Selsam and Han}

\institute{Microsoft Research \email{daselsam@microsoft.com} \\ University of Pittsburgh \email{jmh288@pitt.edu} }

\begin{document}
\maketitle

\begin{abstract}
Lines and circles pose significant scalability challenges in synthetic geometry.
A line with $n$ points implies ${n \choose 3}$ collinearity atoms \icode{coll},
or alternatively, when lines are represented as functions, equality among ${n \choose 2}$ different lines.
Similarly, a circle with $n$ points implies ${n \choose 4}$ cocyclicity atoms \icode{cycl},
or equality among ${n \choose 3}$ circumcircles.
We introduce a new mathematical concept of
\emph{$k$-equivalence relations}, which generalizes equality ($k=1$) and includes both lines ($k=2$) and circles ($k=3$),
and present an efficient proof-producing procedure to compute the closure of a $k$-equivalence relation.
\end{abstract}

\section{Introduction}

Lines and circles pose significant scalability challenges in synthetic geometry.
A line with $n$ points implies ${n \choose 3}$ collinearity atoms \icode{coll},
or alternatively, when lines are represented as functions, equality among ${n \choose 2}$ different lines.
Similarly, a circle with $n$ points implies ${n \choose 4}$ cocyclicity atoms \icode{cycl},
or equality among ${n \choose 3}$ circumcircles.
Although geometry problem statements may not contain any lines or circles with more than a few points on them,
synthetic provers need to introduce auxiliary points, and doing so may result in lines and circles with
dozens if not hundreds of points on them.
To support efficient reasoning in the presence of a large number of auxiliary points, we introduce a new mathematical concept of
\emph{$k$-equivalence relations}, which generalizes equality ($k=1$) to include both lines ($k=2$) and circles ($k=3$),
and present an efficient, proof-producing procedure to compute the closure of a $k$-equivalence relation that
uses exponentially less space (in $k$) than the na\"{i}ve procedure.

\section{$k$-Equivalence Relations}

A binary relation $R$ is an \emph{equivalence} relation provided it satisfies the following three laws:

\begin{center}
\begin{tabular}{|l|l|}
  \hline
  \text{Reflexivity}  & $\forall a . R(a, a)$ \\
  \hline
  \text{Symmetry}     & $\forall a b . R(a, b) \implies R(b, a)$ \\
  \hline
  \text{Transitivity} & $\forall a b c . R(a, b) \wedge R(a, c) \implies R(b, c)$ \\
  \hline
\end{tabular}
\end{center}

In synthetic geometry, the ternary relation $\coll$ representing collinearity
satisfies similar laws:

\begin{center}
\begin{tabular}{|l|l|}
  \hline
  \text{Sub-reflexivity}          & $\forall a b c . \neg \distinct(a, b, c) \implies \coll(a, b, c)$ \\
  \hline
p  \text{Perm-invariance}   & $\forall a b c . \coll(a, b, c) \implies \forall \pi . \permutation(\pi) \implies \coll(\pi(a), \pi(b), \pi(c))$ \\
  \hline
  \text{2-transitivity}           & $\forall a b c d . \coll(a, b, c) \wedge \coll(a, b, d) \wedge a \neq b \implies \coll(b, c, d)$ \\
  \hline
\end{tabular}
\end{center}

as does the quaternary relation $\cycl$ representing cocyclicity:

\begin{center}
\begin{tabular}{|l|l|}
  \hline
  Sub-reflexivity          & $\forall a b c d . \neg \distinct(a, b, c, d) \implies \cycl(a, b, c, d)$ \\
  \hline
  Perm-invariance   & $\forall a b c d . \cycl(a, b, c, d) \implies \forall \pi . \permutation(\pi) \implies \cycl(\pi(a), \pi(b), \pi(c), \pi(d))$ \\
  \hline
  3-transitivity           & $\forall a b c d e . \cycl(a, b, c, d) \wedge \cycl(a, b, c, e) \wedge \distinct(a, b, c) \implies \cycl(b, c, d, e)$ \\
  \hline
\end{tabular}
\end{center}

This pattern leads us to a natural generalization of equivalence relations that we call \emph{$k$-equivalence relations}.
\begin{definition}[$k$-Equivalence Relation] \label{defn:k-equiv}
We say a $(k+1)$-ary relation $R$ is a \emph{$k$-equivalence relation} provided it satisfies the following laws:

\begin{center}
\begin{tabular}{|l|l|}
  \hline
  Sub-reflexivity   & $\forall x_1 \dotsm x_{k+1} . \neg \distinct(\vec{x}) \implies R(\vec{x})$ \\
  \hline
  Perm-invariance   & $\forall x_1 \dotsm x_{k+1} \pi . R(\vec{x}) \wedge \permutation(\pi) \implies R(\pi(\vec{x}))$ \\
  \hline
  $k$-transitivity  & $\forall x_1 \dots x_{k} y_1 y_2 . R(\vec{x}, y_1) \wedge R(\vec{x}, y_2) \wedge \distinct(\vec{x}) \implies R(\vec{x}_{2:}, y_1, y_2)$ \\
  \hline
\end{tabular}
\end{center}
\end{definition}

\section{Basic Properties}

The key property of $k$-equivalence relations that we exploit in our algorithms is that they can be represented compactly in terms of finite sets.

\begin{definition}[$k$-predicate]\label{defn:k-predicate}
  Let \( R \) be a $k$-equivalence relation. Define its $k$-predicate $\Phi_R$ as follows: for any finite set \( X \),
  \[ \Phi_R(X) \coloneqq \bigwedge_{x \subseteq X, |x| = k+1} R(\vec{x}) \]
\end{definition}

\begin{definition}[$k$-predicate laws] \label{defn:k-predicate-laws}
The following laws follow from \Cref{defn:k-equiv} and \Cref{defn:k-predicate}:

\begin{center}
\begin{tabular}{|l|l|}
  \hline
  Sub-reflexivity   & $\forall x, |x| \leq k \implies \Phi_R(x)$ \\
  \hline
  $k$-transitivity  & $\forall x \forall y, \Phi_R(x) \wedge \Phi_R(y) \wedge |x \cap y| \geq k \implies \Phi(x \cup y)$ \\
  \hline
  Projection        & $\forall x y, \Phi_R(x) \wedge y \subseteq x \implies \Phi_R(y)$\\
  \hline
\end{tabular}
\end{center}

\end{definition}

\begin{definition}[$k$-function] A $k$-predicate $\Phi_R$ induces a \emph{$k$-function} $\phi_R$ on sets of size $k$ as follows:
  \[ \phi_R(x) \coloneqq \{ y : \Phi_R(x \cup y) \} \]
\end{definition}

\begin{example}
Note that $\phi_\coll(\{a, b\})$ is the set of points collinear with \( \{ a, b \} \), \ie the line through \( a \) and \( b \).
Similarly, $\phi_\cycl(\{a, b, c\})$ is the set of points cocyclic with \( \{ a, b, c \}\), \ie the circumcircle of \(a, b, c \).
\end{example}

\begin{lemma}
  Let $x_1, x_2$ be two sets of size $k$. Then
  \[ \phi_R(x_1) = \phi_R(x_2) \iff \Phi_R(x_1 \cup x_2) \]
\end{lemma}

\begin{proof}
  First suppose \( \{ y : \Phi_R(x_1 \cup y) \} = \{ y : \Phi_R(x_2 \cup y) \} \). Then $x_2$ is in the second set by subreflexivity so it must be in the first set,
  which yields $\Phi_R(x_1 \cup x_2)$.
  Now suppose \( \Phi_R(x_1 \cup x_2) \), and let $y$ be such that \( \Phi_R(x_1 \cup y) \). Since \( x_1 \cup y \) and \( x_1 \cup x_2 \) overlap at \( x_1 \) of size $k$,
  it follows by $k$-transitivity that \( \Phi_R(x_1 \cup x_2 \cup y ) \), and by projection that \( \Phi_R(x_2 \cup y) \).
\end{proof}

\section{Deciding $k$-Equivalence Relations}\label{sec:kdecide}

It is well known that a traditional equivalence relation defines a partition, which can be represented as a set of disjoint sets
and computed using \eg the \emph{union-find} algorithm~\cite{galler1964improved}. Similarly, a $k$-equivalence relation can be represented
as a set of sets whose pairwise intersections are less than $k$. When $k=1$, the sets must be disjoint, but this is not the case for higher values of $k$.
For example, a point $A$ in the plane may be collinear with two points $B_1$ and $C_1$, and also with two points $B_2$ and $C_2$, but this
does not imply that the five points are all collinear.

\subsection{Procedure}\label{sec:kdecide:subsec:procedure}

Fix a $k$-equivalence relation $R$, and consider a sequence $H_i$ of $R$-atoms over a set of \emph{distinct} terms $x_i$
(we relax the assumption of distinctness in \Cref{sec:kdecide:subsec:nondistinct}).
We are interested in determining, for a given $R$-atom $R(y_1, \dotsc, y_{k+1})$, whether or not
$\bigwedge_i H_i \implies R(y_1, \dotsc, y_{k+1})$.
The main idea of our procedure is to reason using the $k$-predicate $\Phi_R$, and to saturate with the rules from \Cref{defn:k-predicate-laws}
rather than those of \Cref{defn:k-equiv}.
Our procedure produces proofs using the following constructors:
\begin{enumerate}
\item \icode{assume(<hypothesis-idx>)}
\item \icode{subrefl(<terms>)}
\item \icode{trans(<pf1>, <pf2>)}
\item \icode{project(<pf>, <terms>)}
\end{enumerate}
where the latter three correspond to the laws of \Cref{defn:k-predicate-laws}.
Define a \emph{$k$-set} $s$ to be a set representing the fact $\Phi_R(s)$.
Our procedure maintains three datastructures:

\begin{enumerate}
\item \icode{ksets}: an array of $k$-sets.
  Note that we store the entire history of $k$-sets to facilitate proof production
  but only some $k$-sets are considered \emph{active}.
\item \icode{proofs}: an array of \emph{proofs}, one for each $k$-set.
\item \icode{term2parents}: a map from terms to the IDs of the \emph{active} ksets that contain it.
\end{enumerate}

\Cref{fig:proc:KDecide} shows pseudocode for the procedure. We first initialize
\icode{ksets},
\icode{proofs},
and \icode{term2parents} to empty (L\ref{proc:KDecide:init}).
We then iterate over the hypotheses in sequence (L\ref{proc:KDecide:for}).
For the $i$th hypothesis \icode{R(xs)}, we first create a new $k$-set for it
using \Call{New}{} (L\ref{proc:KDecide:call:new}), which involves appending new entries to \icode{ksets} and \icode{proofs}
and updating the \icode{term2parents} datastructure (L\ref{proc:KDecide:update-term2parents}).
At this point, the new $k$-set may overlap one or more existing $k$-sets by $k$ elements.
Thus, we need to detect these overlaps and merge $k$-sets as necessary.
We accomplish by calling \Call{FindMerges}{} on the newly created $k$-set (L\ref{proc:KDecide:call:findMerges}).
To find the necessary merges, we compute the $k$-sets that overlap at least $k$ with the newly created $k$-set
by first multiset-unioning the parents of the new $k$-set (L\ref{proc:KDecide:findMerges:union})
and then filtering the parents that occur atleast $k$ times (L\ref{proc:KDecide:findMerges:filter}).
If there are no such $k$-sets (besides the new one), there is nothing to do (L\ref{proc:KDecide:findMerges:check1}).
Otherwise we fold over the filtered $k$-sets, merging them (with \Call{Merge}{}) in sequence into one big $k$-set (L\ref{proc:KDecide:findMerges:merge}).
The procedure \Call{Merge}{} simply deactivates the old $k$-sets by removing them from \icode{term2parents} (L\ref{proc:KDecide:merge:deactivate1}-\ref{proc:KDecide:merge:deactivate2})
and then creates new $k$-set with the union of the two old $k$-sets (L\ref{proc:KDecide:merge:new}).
Note that once the matches have been merged into one big $k$-set, this new $k$-set may overlap atleast $k$ with another active $k$-set;
thus we must recursively find the merges of the new $k$-set (L\ref{proc:KDecide:findMerges:findMerges}).
Finally, to answer the query \icode{R(xs)}, we intersect the parents of the elements of \icode{xs} (L\ref{proc:KDecide:resolveQuery});
if the result is empty the query is not entailed (L\ref{proc:KDecide:matches-empty}), otherwise we call \Call{explain}{} to produce a proof of the query,
which we discuss in \Cref{sec:kdecide:subsec:proofs}. Note that we can easily extend \Call{KDecide}{} to support arbitrary-sized queries by first
checking if the query \icode{R(xs)} has size $\leq k$ and if so returning \icode{subrefl(xs)}.

\begin{figure}
\small
\begin{algorithmic}[1]
  \Procedure{KDecide}{\icode{hyps}, \icode{query}} \label{proc:KDecide}
  \State \icode{ksets}, \icode{proofs}, \icode{term2parents} $\gets$ \{\}, \{\}, \{\} \label{proc:KDecide:init}

  \For{i, $R$(\icode{xs}) in \icode{hyps}} \label{proc:KDecide:for}
    \State \icode{n} $\gets$  \Call{New}{\icode{xs}, \icode{assume(i)}}  \label{proc:KDecide:call:new}
    \State \Call{FindMerges}{\icode{n}} \label{proc:KDecide:call:findMerges}
  \EndFor

  \State $R$(\icode{xs}) $\gets$ \icode{query}
  \State \icode{matches} $\gets$ intersect parents of \icode{xs} \label{proc:KDecide:resolveQuery}
  \If{\icode{matches.isEmpty}} \Return \icode{not-entailed} \label{proc:KDecide:matches-empty}
  \Else \,  \Return \Call{explain}{\icode{matches[0]}, \icode{xs}} \label{proc:KDecide:explain}
  \EndIf
  \EndProcedure

  \Procedure{New}{\icode{xs}, \icode{proof}}
  \State \icode{n} $\gets$ \icode{ksets.size()}
  \State \icode{ksets.append(xs)}
  \State \icode{proofs.append(proof)}
  \For{\icode{x} in \icode{xs}}
  \icode{term2parent[x].insert(n)} \label{proc:KDecide:update-term2parents}
  \EndFor
  \State \Return \icode{n}
  \EndProcedure

  \Procedure{FindMerges}{\icode{n}}
  \State \icode{allParents} $\gets$ multiset union of the parents of \icode{ksets[n]} \label{proc:KDecide:findMerges:union}
  \State \icode{matches} $\gets$ subset of \icode{allParents} that appear at least $k$ times \label{proc:KDecide:findMerges:filter}
  \If{\icode{|matches| < 2}} \Return \EndIf \label{proc:KDecide:findMerges:check1}
  \State \icode{last} $\gets$ \icode{matches[0]}
  \For{$i$ = 1 to \icode{length}(\icode{matches})-1}
  \icode{last} $\gets$ \Call{Merge}{\icode{last}, \icode{matches[i]}}
  \EndFor \label{proc:KDecide:findMerges:merge}
  \State \Call{FindMerges}{\icode{last}} \label{proc:KDecide:findMerges:findMerges}
  \EndProcedure

  \Procedure{Merge}{\icode{i1}, \icode{i2}}
  \For{\icode{x} in \icode{ksets[i1]}}
  \icode{term2parent[x].remove(i1)}
  \EndFor \label{proc:KDecide:merge:deactivate1}
  \For{\icode{x} in \icode{ksets[i2]}}
  \icode{term2parent[x].remove(i2)}
  \EndFor \label{proc:KDecide:merge:deactivate2}
  \State \Return \Call{New}{\icode{union(ksets[i1], ksets[i2])}, \icode{trans(i1,i2)}}  \label{proc:KDecide:merge:new}
  \EndProcedure
\end{algorithmic}
 \caption{Deciding $k$-equivalence relations.}
 \label{fig:proc:KDecide}
 \end{figure}

 \subsection{Example}\label{sec:kdecide:subsec:example}

 We provide more intuition for our procedure by walking through the following small example ($k=2$):
 \[ R(a, b, c) \wedge R(c, d, e) \wedge R(e, f, g) \wedge R(a, d, g) \wedge R(b, c, d) \]
 The final state of the procedure is shown in \Cref{table:example:table}. None of the first four hypotheses trigger any merges.
 The fifth hypothesis (\icode{H4}) matches with both \icode{H0} and \icode{H1}, producing the $k$-set $\{a, b, c, d, e \}$.
 This $k$-set then (recursively) matches with \icode{H3} yielding $\{a, b, c, d, e, g \}$ which (recusively)
 matches with \icode{H2}, yielding a singleton $k$-set containing the union of all original hypotheses.
 \begin{table}
  \begin{center}
   \begin{tabular}{|l|l|l|l|}
   \hline
     \textbf{KSet Index} & \textbf{Active} & \textbf{Proof} & \textbf{Terms} \\ \hline
     0 & 0 & assume(H0) & a, b, c \\ \hline
     1 & 0 & assume(H1) & c, d, e \\ \hline
     2 & 0 & assume(H2) & e, f, g \\ \hline
     3 & 0 & assume(H3) & a, d, g \\ \hline
     4 & 0 & assume(H4) & b, c, d \\ \hline
     5 & 0 & trans(0, 4)    & a, b, c, d \\ \hline
     6 & 0 & trans(1, 5)    & a, b, c, d, e \\ \hline
     7 & 0 & trans(3, 6)    & a, b, c, d, e, g \\ \hline
     8 & 1 & trans(2, 7)    & a, b, c, d, e, f, g \\ \hline
   \end{tabular}
   \end{center}
   \caption{The state of the procedure after asserting the hypotheses in the example problem of
     \Cref{sec:kdecide:subsec:example} given by \( R(a, b, c) \wedge R(c, d, e) \wedge R(e, f, g) \wedge R(a, d, g) \wedge R(b, c, d) \).}
   \label{table:example:table}
\end{table}

 \subsection{Producing proofs} \label{sec:kdecide:subsec:proofs}

 Suppose for a given query $R(\vec{x})$, we find a $k$-set $y$ containing $\vec{x}$.
 We can easily produce a proof of $R(\vec{x})$ by simply projecting the proof of $\Phi_R(y)$ stored in the \icode{proofs} array
 with the \icode{project} proof constructor corresponding to the project rule in \Cref{defn:k-predicate-laws}.
 However, this proof may be extremely suboptimal in general.
 We take inspiration from~\cite{nieuwenhuis2005proof} and provide an \Call{Explain}{} operation that produces a proof of the query using a minimal subset of the hypotheses.

 \Cref{fig:proc:Explain} contains pseudocode for the \Call{Explain}{} procedure, which is called on the index $n$ of the $k$-set that contains the query,
 along with the terms in the query.  The procedure rests on two key insights. First, if we want to produce a proof of \icode{xs} from a $k$-set constructed from a merge,
 and if one of the merged $k$-sets \icode{s1} already contains \icode{xs}, then we can produce a proof directly from \icode{s1}
 without considering \icode{s2} (L\ref{proc:explain:insight1a}-\ref{proc:explain:insight1b}).
 Second, even if \icode{xs} is not contained in either \icode{s1} or \icode{s2},
 we still do not need to produce proofs of \icode{s1} and \icode{s2} in their entirety;
 if we can produce child proofs of \icode{union(inter(s1, s2), inter(s1, xs))} and \icode{union(inter(s1, s2), inter(s2, xs))} respectively,
 we can glue them together with $k$-transitivity to produce a proof of a set containing \icode{xs}, from which we can project a proof of \icode{xs}.

\begin{figure}
\small
\begin{algorithmic}[1]
  \Procedure{Explain}{\icode{n}, \icode{xs}} \label{proc:Explain}
  \If{\icode{histories[n] = assume(i)}}
  \Return \icode{assume(i)}
  \ElsIf{\icode{histories[n] = merge(i1, i2)}}
  \If{\icode{xs} $\subseteq$ \icode{ksets[i1]}} \Return \Call{Explain}{\icode{i1}, \icode{xs}} \label{proc:explain:insight1a}
  \ElsIf{\icode{xs} $\subseteq$ \icode{ksets[i2]}} \Return \Call{Explain}{\icode{i2}, \icode{xs}} \label{proc:explain:insight1b}
  \Else
  \State \icode{anchor} $\gets$ \icode{ksets[i1]} $\cap$ \icode{ksets[i2]}
  \State \icode{pf1} $\gets$ \Call{Explain}{i1, \icode{anchor} $\cup$ (\icode{ksets[i1]} $\cap$ \icode{xs})}
  \State \icode{pf2} $\gets$ \Call{Explain}{i2, \icode{anchor} $\cup$ (\icode{ksets[i2]} $\cap$ \icode{xs})}
  \State \Return \icode{project(trans(pf1, pf2), xs)}
  \EndIf
  \EndIf
  \EndProcedure
\end{algorithmic}
 \caption{Producing compact proofs.}
 \label{fig:proc:Explain}
\end{figure}

\paragraph{Example.} Suppose that after asserting the hypotheses of \Cref{sec:kdecide:subsec:example} we queried \Call{Explain}{8, \icode{\{a, b, d\}}}.
\Call{Explain}{} will take the branch in L\ref{proc:explain:insight1b} three consecutive times and finally return the
minimal proof \icode{project(trans(assume(H0), assume(H4)), \{a, b, d\})}.

\subsection{Terms that may not be distinct} \label{sec:kdecide:subsec:nondistinct}

We now show how to lift the simplifying assumption of \ref{sec:kdecide:subsec:procedure} that each of the terms that
appeared in the $k$-equivalence relation are known to be distinct.
Suppose that rather than assuming that every term is distinct, we assume that we have a partition $\sigma$ on terms such that two terms
in different equivalence classes $\gamma$ are known to be distinct.
In geometry solvers, it is common to use numericals diagrams to determine acceptable
distinctness conditions to assume, in which case such a partition can be determined by rounding the coordinates of each (possibly equal) point to a given precision.
To support this scenario, it suffices to tweak L\ref{proc:KDecide:findMerges:union}
by first \emph{set}-unioning the parents \emph{within} each equivalence class before \emph{multiset}-unioning the results.
This will compute the set of $k$-sets that overlap with the new $k$-set at atleast $k$ points that are known to be distinct.

\subsection{Asymptotics} \label{sec:kdecide:subsec:asymptotics}

Let $n$ be the number of hypotheses and $m$ the maximum number of parents of any term at any time during the procedure.

\begin{lemma}
There can only ever be as many as $n$ active $k$-sets.
\end{lemma}

\begin{proof}
  Only processing a new hypothesis increases the number of active $k$-sets;
  all other calls to \Call{New}{} are during merges,
  which first deactivate two $k$-sets and so decrease the number of active $k$-sets by one.
\end{proof}

\begin{lemma}
\( m = O(n) \)
\end{lemma}

\begin{lemma}
  There can be at most $n-1 = O(n)$ merges.
\end{lemma}

\begin{lemma}
  \Call{FindMerges}{} is called at most $2n = O(n)$ times.
\end{lemma}

\begin{lemma}
The largest size of any $k$-set is $k+n$.
\end{lemma}

\begin{proof}
  Any $k$-set must be the union of some $n_0$-element subset of the set of the $n$ original $k+1$-element $k$-sets corresponding to the hypotheses.
  For these $n_0$ $k$-sets to be merged into one, they all must overlap $k$ with their union; thus the union can have at most $n_0+k$ elements.
  The result follows from $n_0 \leq n$.
\end{proof}

\begin{lemma}
  The cumulative execution of time of L\ref{proc:KDecide:findMerges:union} is at most $2n(k+n)m = O(knm + n^2m)$ time.
\end{lemma}

\begin{lemma}
  The cumulative execution time of L\ref{proc:KDecide:update-term2parents}, L\ref{proc:KDecide:merge:deactivate1}, L\ref{proc:KDecide:merge:deactivate2}, L\ref{proc:KDecide:merge:new}
  are all bounded by \( O(n(k+n)) \).
\end{lemma}

\begin{theorem}[Worst-Case Upper Bound]
  The worst-case running time of \Call{KDecide}{\icode{hyps}, \icode{query}}
  is $O \left( kmn + n^2m \right) $. Since $m = O(n)$, this simplifies to $O\left( kn^2 + n^3 \right)$.
\end{theorem}

\begin{corollary}
  The worst-case running time is linear in $k$, whereas the na\"{i}ve implementation is exponential in $k$.
\end{corollary}

\subsection{Integrating with congruence closure}\label{sec:kdecide:subsec:kcc}

It is relatively straightforward to integrate the decision procedure of \Cref{sec:kdecide:subsec:procedure} into a congruence closure procedure.
Suppose we have a sequence of equalities and $R$-atoms for various $k$-equivalence relations, and for simplicity assume that
no equalities among $k$-functions are provided explicitly.
For a given $k$-equivalence relation, we can represent the entailed $R$-atoms compactly using $k$-sets, where each $k$-set also stores
its canonical $k$-function application. Whenever two equivalence classes are merged, in addition to the standard congruence closure bookkeeping, we traverse all $k$-sets that include
any member of the smaller class. For each one, we replace all terms with their new representatives and call \Call{New}{} on the result.
Whenever two $k$-sets are merged, we also merge the $E$-classes of their canonical $k$-function applications.

\section{Related Work}\label{sec:relatedwork}

In \cite{chou2000deductive}, both lines and circles are represented using lists of points,
and the corresponding permutation and transitivity rules are claimed to be built-in to the solver;
however, few details are provided about how these rules are propagated.
Equivalence relations have previously been generalized to \emph{ternary equivalence relations}~\cite{rainich1952ternary},
which include collinearity but does not generalize to cocyclicity.
The concept of \emph{$E$-sequences}~\cite{szmielew1981n} is related to our notion of $k$-equivalence relation as follows:
a $(k+1)$-ary relation $R$ is a $k$-equivalence relation if and only if the sequence of relations \( ( \mathrm{notDistinct}_2, \dotsc, \mathrm{notDistinct_k}, R ) \) forms an $E$-sequence.

\section*{Acknowledgments}

We thank Nikolaj Bj{\o}rner for detailed discussions.


\newpage
\bibliography{ksets}

\begin{thebibliography}{1}
\providecommand{\url}[1]{\texttt{#1}}
\providecommand{\urlprefix}{URL }
\providecommand{\doi}[1]{https://doi.org/#1}

\bibitem{chou2000deductive}
Chou, S.C., Gao, X.S., Zhang, J.Z.: A deductive database approach to automated
  geometry theorem proving and discovering. Journal of Automated Reasoning
  \textbf{25}(3),  219--246 (2000)

\bibitem{galler1964improved}
Galler, B.A., Fisher, M.J.: An improved equivalence algorithm. Communications
  of the ACM  \textbf{7}(5),  301--303 (1964)

\bibitem{nieuwenhuis2005proof}
Nieuwenhuis, R., Oliveras, A.: Proof-producing congruence closure. In:
  International Conference on Rewriting Techniques and Applications. pp.
  453--468. Springer (2005)

\bibitem{rainich1952ternary}
Rainich, G., et~al.: Ternary relations in geometry and algebra. The Michigan
  Mathematical Journal  \textbf{1}(2),  97--111 (1952)

\bibitem{szmielew1981n}
Szmielew, W.: On n-ary equivalence relations and their application to geometry
  (1981)

\end{thebibliography}

\end{document}